%% file: main.tex
\documentclass[runningheads]{llncs}
\usepackage[T1]{fontenc}
\usepackage{graphicx}
\usepackage{tabularx}
\usepackage{array}
\usepackage{siunitx}
\usepackage{makecell}
\usepackage{microtype}
\usepackage{cite}

\usepackage[T1]{fontenc}
\usepackage{graphicx}
\usepackage{subcaption}
\usepackage{hyperref}
\usepackage{listings, listings-rust}
\usepackage{amsmath,amsfonts}
\usepackage{algorithm}
\usepackage{algpseudocode}
\usepackage{booktabs}
\usepackage{makecell,multirow}
\usepackage{hyperref}
\usepackage{tcolorbox}
\usepackage[htt]{hyphenat} 
\usepackage{xspace}
\usepackage{doi}
\usepackage{xcolor}
\usepackage{wrapfig}
\usepackage{enumitem}

\lstdefinestyle{gogo}{
  language=Go,
  basicstyle=\ttfamily\small,
  keywordstyle=\bfseries,
  commentstyle=\itshape,
  stringstyle=\itshape,
  numbers=left,
  numberstyle=\scriptsize,
  stepnumber=1,
  numbersep=8pt,
  showstringspaces=false,
  breaklines=true,
  columns=fullflexible,
  captionpos=b,
  aboveskip=8pt,
  belowskip=8pt,
}

\lstdefinestyle{rusty}{
  language=Rust,
  basicstyle=\ttfamily\small,
  keywordstyle=\bfseries,
  commentstyle=\itshape,
  stringstyle=\itshape,
  numbers=left,
  numberstyle=\scriptsize,
  stepnumber=1,
  numbersep=8pt,
  showstringspaces=false,
  breaklines=true,
  columns=fullflexible,
  captionpos=b,
  aboveskip=8pt,
  belowskip=8pt,
}

\lstdefinelanguage{Proto}{
  morekeywords={
    syntax,import,package,option,message,enum,service,rpc,returns,
    map,oneof,reserved,repeated,optional,bytes,string,int32,int64,
    uint32,uint64,sint32,sint64,fixed32,fixed64,sfixed32,sfixed64,bool,
    float,double
  },
  sensitive=true,
  morecomment=[l]{//},
  morecomment=[s]{/*}{*/},
  morestring=[b]{"},
}

\lstdefinestyle{proto}{
  language=Proto,
  basicstyle=\ttfamily\small,
  keywordstyle=\bfseries,
  commentstyle=\itshape,
  stringstyle=\itshape,
  numbers=left,
  numberstyle=\scriptsize,
  stepnumber=1,
  numbersep=8pt,
  showstringspaces=false,
  breaklines=true,
  columns=fullflexible,
  captionpos=b,
  aboveskip=8pt,
  belowskip=8pt,
}

\definecolor{GoBlue}{RGB}{0,102,204}
\definecolor{RustOrange}{RGB}{204,102,0}

\definecolor{Rust}{RGB}{249,104,21}
\definecolor{Go}{RGB}{78,196,234}
\definecolor{prompt_bg}{RGB}{252,255,221}
\definecolor{prompt_title}{RGB}{0,51,102}

\newif\ifcomments

\commentstrue

\ifcomments
\providecommand{\hlz}[1]{\textbf{\textcolor{blue}{HLZ: #1}}}
\providecommand{\cd}[1]{\textbf{\textcolor{red}{CD: #1}}}
\providecommand{\as}[1]{\textbf{\textcolor{green}{AS: #1}}}

\else
\providecommand{\hlz}[1]{}
\providecommand{\cd}[1]{}
\providecommand{\as}[1]{}
\fi

\input{macro}

\begin{document}
\title{Validated Code Translation for Projects with External Libraries}
%
%

\author{
Hanliang Zhang\inst{1,2} \and
Arindam Sharma\inst{2} \and
Cristina David\inst{2} \and
Meng Wang\inst{2} \and
Brandon Paulsen\inst{1} \and
Daniel Kroening\inst{1} \and
Wenjia Ye\inst{2} \and
Taro Sekiyama\inst{3}
}

\authorrunning{Zhang et al.}

\institute{
Amazon Inc., UK\\
\email{\{hlz,bpaulse,dkr\}@amazon.co.uk}
\and
University of Bristol, UK\\
\email{\{arindam.sharma,cristina.david,meng.wang,wenjia.ye\}@bristol.ac.uk}
\and
National Institute of Informatics, Japan\\
\email{sekiyama@nii.ac.jp}
}

%
%
\maketitle              
\begin{abstract} 

Large Language Models (LLMs) have shown promise for program translation, particularly for migrating systems code to memory-safe languages such as Rust. However, existing approaches struggle when source programs depend on external libraries: LLMs frequently hallucinate non-existent target APIs and fail to generate call-enabling imports; moreover, validating semantic equivalence is challenging when the code manipulates opaque, library-defined types. We present a \emph{translation and validation framework for translating Go projects with external dependencies to Rust}. Our approach combines (i) a retrieval mechanism that maps Go library APIs to Rust APIs, and (ii) a cross-language validation pipeline that establishes language interoperability in the presence of opaque library types by synthesising adapters exclusively from public library APIs, 
prior to validating I/O equivalence. We evaluate our system on six real-world Go repositories with non-trivial external dependencies.
Our approach significantly increases both the compilation and equivalence success rate (up to 100\% in the most dependency-heavy case; approx. 2x on average) by enabling validated translation that manipulate opaque, library-defined types.

\keywords{Code translation \and  Rust \and Go \and Large Language Model.}
\end{abstract}
\input{sections/intro-brandon-1}
\input{sections/overview}
\input{sections/knowledge-base}
\input{sections/validate}
\input{sections/experiment}

\input{sections/related_work}
\input{sections/conclusions}

%
%
%
%
\bibliographystyle{splncs04}
\bibliography{reference-combined}

 \appendix
 \input{sections/api-index-construction}

 \input{sections/appendix}
\end{document}

%% file: macro.tex
\newcommand{\oxidizer}{\textbf{Oxidizer}}
\newcommand{\liboxidizer}{\textbf{CrossCrate}}

\newcommand{\go}[1]{\textcolor{Go}{#1}}
\newcommand{\rust}[1]{\textcolor{Rust}{#1}}

\newcommand{\rustcode}[1]{\texttt{\color{Rust}\small#1}}
\newcommand{\gocode}[1]{\texttt{\color{Go}\small#1}}

\newcommand{\rustcrate}[1]{\textbf{\color{Rust}\small#1}}
\newcommand{\gopackage}[1]{\textbf{\color{Go}\small#1}}

\newcommand{\elem}[1]{\mathit{\MakeLowercase{#1}}}

\newcommand{\golang}{\textsc{Go}\xspace}
\newcommand{\rustlang}{\textsc{Rust}\xspace}
\newcommand{\protobuf}{\textsc{Protobuf}\xspace}
\newcommand{\gocloc}{\textsc{gocloc}\xspace}

\newcommand{\api}{API} 

\newcommand{\itemsmap}{\textit{ApiIndex}}

\newcommand{\age}{\textsc{age}\xspace}
\newcommand{\ageStars}{21.1k\xspace}
\newcommand{\ageForks}{\num{609}\xspace}
\newcommand{\ageLoC}{\num{51}\xspace}

\newcommand{\easyScrypt}{\textsc{easy-scrypt}\xspace}
\newcommand{\easyScryptStars}{\num{25}\xspace}
\newcommand{\easyScryptForks}{\num{3}\xspace}
\newcommand{\easyScryptLoC}{\num{196}\xspace}

\newcommand{\simpleScrypt}{\textsc{simple-scrypt}\xspace}
\newcommand{\simpleScryptStars}{\num{201}\xspace}
\newcommand{\simpleScryptForks}{\num{27}\xspace}
\newcommand{\simpleScryptLoC}{\num{322}\xspace}

\newcommand{\duration}{\textsc{duration}\xspace}
\newcommand{\durationStars}{\num{32}\xspace}
\newcommand{\durationForks}{\num{2}\xspace} 
\newcommand{\durationLoC}{\num{166}\xspace} 

\newcommand{\argonTwoHashing}{\textsc{argon2-hashing}\xspace}
\newcommand{\argonTwoHashingStars}{\num{25}\xspace} 
\newcommand{\argonTwoHashingForks}{\num{4}\xspace}  
\newcommand{\argonTwoHashingLoC}{\num{498}\xspace}  

\newcommand{\pbkdft}{\textsc{pbkdf2}\xspace}
\newcommand{\pbkdftStars}{\num{11}\xspace}
\newcommand{\pbkdftForks}{\num{5}\xspace}
\newcommand{\pbkdftLoC}{\num{270}\xspace}

\newcommand{\pbkdftRAGLoad}{\num{341}\xspace}
\newcommand{\pbkdftIOPercent}{\num{100}\xspace}
\newcommand{\pbkdftIOPercentLibOnly}{\num{100}\xspace}

\newcommand{\pbkdftCompPercent}{\num{100}\xspace}
\newcommand{\pbkdftCompPercentLibOnly}{\num{100}\xspace}

\newcommand{\easyScryptRAGLoad}{\num{1382}\xspace}
\newcommand{\easyScryptIOPercent}{\num{100}\xspace}
\newcommand{\easyScryptIOPercentLibOnly}{\num{100}\xspace}

\newcommand{\easyScryptCompPercent}{\num{100}\xspace}
\newcommand{\easyScryptCompPercentLibOnly}{\num{100}\xspace}

\newcommand{\simpleScryptRAGLoad}{\num{1388}\xspace}
\newcommand{\simpleScryptIOPercent}{\num{100}\xspace}
\newcommand{\simpleScryptIOPercentLibOnly}{\num{100}\xspace}

\newcommand{\simpleScryptCompPercent}{\num{100}\xspace}
\newcommand{\simpleScryptCompPercentLibOnly}{\num{100}\xspace}

\newcommand{\durationRAGLoad}{\num{1341}\xspace}
\newcommand{\durationIOPercent}{\num{100}\xspace}
\newcommand{\durationIOPercentLibOnly}{\num{100}\xspace}

\newcommand{\durationCompPercent}{\num{100}\xspace}
\newcommand{\durationCompPercentLibOnly}{\num{100}\xspace}

\newcommand{\ageRAGLoad}{\num{3543}\xspace}
\newcommand{\ageIOPercent}{\num{80}\xspace}
\newcommand{\ageIOPercentLibOnly}{\num{80}\xspace}
\newcommand{\ageCompPercent}{\num{75}\xspace}
\newcommand{\ageCompPercentLibOnly}{\num{75}\xspace}
\newcommand{\ageNoRAGIOPercent}{\num{0}\xspace}
\newcommand{\ageNoRAGIOPercentLibOnly}{\num{0}\xspace}
\newcommand{\ageNoRAGCompPercent}{\num{0}\xspace}
\newcommand{\ageNoRAGCompPercentLibOnly}{\num{0}\xspace}
\newcommand{\ageRAGNoImportIOPercent}{\num{0}\xspace}
\newcommand{\ageRAGNoImportIOPercentLibOnly}{\num{0}\xspace}
\newcommand{\ageRAGNoImportCompPercent}{\num{0}\xspace}
\newcommand{\ageRAGNoImportCompPercentLibOnly}{\num{0}\xspace}

\newcommand{\argonTwoHashingRAGLoad}{\num{139}\xspace}
\newcommand{\argonTwoHashingIOPercent}{\num{100}\xspace}
\newcommand{\argonTwoHashingIOPercentLibOnly}{\num{100}\xspace}

\newcommand{\argonTwoHashingCompPercent}{\num{100}\xspace}
\newcommand{\argonTwoHashingCompPercentLibOnly}{\num{100}\xspace}

\newcommand{\pbkdftNoRAGIOPercent}{\num{0}\xspace}
\newcommand{\pbkdftNoRAGIOPercentLibOnly}{\num{0}\xspace}
\newcommand{\pbkdftNoRAGCompPercent}{\num{0}\xspace}
\newcommand{\pbkdftNoRAGCompPercentLibOnly}{\num{0}\xspace}

\newcommand{\pbkdftRAGNoImportIOPercent}{\num{0}\xspace}
\newcommand{\pbkdftRAGNoImportIOPercentLibOnly}{\num{0}\xspace}
\newcommand{\pbkdftRAGNoImportCompPercent}{\num{0}\xspace}
\newcommand{\pbkdftRAGNoImportCompPercentLibOnly}{\num{0}\xspace}

\newcommand{\easyScryptNoRAGIOPercent}{\num{33}\xspace}
\newcommand{\easyScryptNoRAGIOPercentLibOnly}{\num{0}\xspace}
\newcommand{\easyScryptNoRAGCompPercent}{\num{33}\xspace}
\newcommand{\easyScryptNoRAGCompPercentLibOnly}{\num{0}\xspace}

\newcommand{\easyScryptRAGNoImportIOPercent}{\num{33}\xspace}
\newcommand{\easyScryptRAGNoImportIOPercentLibOnly}{\num{0}\xspace}
\newcommand{\easyScryptRAGNoImportCompPercent}{\num{33}\xspace}
\newcommand{\easyScryptRAGNoImportCompPercentLibOnly}{\num{0}\xspace}

\newcommand{\simpleScryptNoRAGIOPercent}{\num{60}\xspace}
\newcommand{\simpleScryptNoRAGIOPercentLibOnly}{\num{0}\xspace}
\newcommand{\simpleScryptNoRAGCompPercent}{\num{60}\xspace}
\newcommand{\simpleScryptNoRAGCompPercentLibOnly}{\num{0}\xspace}

\newcommand{\simpleScryptRAGNoImportIOPercent}{\num{60}\xspace}
\newcommand{\simpleScryptRAGNoImportIOPercentLibOnly}{\num{0}\xspace}
\newcommand{\simpleScryptRAGNoImportCompPercent}{\num{60}\xspace}
\newcommand{\simpleScryptRAGNoImportCompPercentLibOnly}{\num{0}\xspace}

\newcommand{\argonTwoHashingNoRAGIOPercent}{\num{80}\xspace}
\newcommand{\argonTwoHashingNoRAGIOPercentLibOnly}{\num{0}\xspace}
\newcommand{\argonTwoHashingNoRAGCompPercent}{\num{80}\xspace}
\newcommand{\argonTwoHashingNoRAGCompPercentLibOnly}{\num{0}\xspace}

\newcommand{\argonTwoHashingRAGNoImportIOPercent}{\num{80}\xspace}
\newcommand{\argonTwoHashingRAGNoImportIOPercentLibOnly}{\num{0}\xspace}
\newcommand{\argonTwoHashingRAGNoImportCompPercent}{\num{80}\xspace}
\newcommand{\argonTwoHashingRAGNoImportCompPercentLibOnly}{\num{0}\xspace}

\newcommand{\durationNoRAGIOPercent}{\num{75}\xspace}
\newcommand{\durationNoRAGIOPercentLibOnly}{\num{0}\xspace}
\newcommand{\durationNoRAGCompPercent}{\num{75}\xspace}
\newcommand{\durationNoRAGCompPercentLibOnly}{\num{0}\xspace}

\newcommand{\durationRAGNoImportIOPercent}{\num{75}\xspace}
\newcommand{\durationRAGNoImportIOPercentLibOnly}{\num{0}\xspace}
\newcommand{\durationRAGNoImportCompPercent}{\num{75}\xspace}
\newcommand{\durationRAGNoImportCompPercentLibOnly}{\num{0}\xspace}

\newcommand{\avgRAGLoad}{\num{1356}}

\newcommand{\avgFullComp}{\num{95.83}}
\newcommand{\avgFullCompLibOnly}{\num{95.83}}

\newcommand{\avgFullIO}{\num{95.83}}
\newcommand{\avgFullIOLibOnly}{\num{95.83}}

\newcommand{\avgNoRAGComp}{\num{41.33}}
\newcommand{\avgNoRAGCompLibOnly}{\num{0}}

\newcommand{\avgNoRAGIO}{\num{41.33}}
\newcommand{\avgNoRAGIOLibOnly}{\num{0}}

\newcommand{\avgNoImportComp}{\num{41.33}}
\newcommand{\avgNoImportCompLibOnly}{\num{0}}

\newcommand{\avgNoImportIO}{\num{41.33}}
\newcommand{\avgNoImportIOLibOnly}{\num{0}}

%% file: sections/intro-brandon-1.tex
\section{Introduction}

Program translation is an important task with many applications, and has
received significant attention from academia~\cite{wang_program_2025,
wang_repotransbench_2024,shiraishi_smartc2rust_2026},
industry~\cite{eniser_towards_2024,ibrahimzada_alphatrans_2025,roziere_unsupervised_2020},
and government
organizations~\cite{noauthor_eliminating_nodate,darpa_tractor_2025}. 
Initial work in this area was done in the context of
compilers~\cite{pnueli_translation_1998,necula_translation_2000}, but
recently source-to-source translation has gained popularity.  Along this
line, many approaches develop rule-based
systems~\cite{noauthor_sharpen_nodate,noauthor_c_nodate,malabarba_mohca-java_1999,buddrus_cappuccinoc_1998,immunant_c2rust_nodate,zhang_ownership_2023,sharp_corrode_nodate,ling_rust_2022,hong_dont_2024,fromherz_compiling_2024,emre_translating_2021,zhang_heterogen_2022,zhan_verified_2024,wang_user-customizable_2023,qiu_tenspiler_2024,kamil_verified_2016,ahmad_automatically_2019},
however these translations are generally unidiomatic, especially for target
languages like Rust~\cite{pan_lost_2024,eniser_towards_2024}.  More recent
approaches have combined LLMs, which generate idiomatic but untrusted
translations, with semantic validation systems, such as differential
testing~\cite{zhang_scalable_2025,shetty_syzygy_2024,zhou_llm-driven_2025,ibrahimzada_alphatrans_2025,abid_gluetest_2024,ke_advancing_2025,bai2025rustassure,li_adversarial_2025},
fuzzing~\cite{eniser_automatically_2024,li_adversarial_2025}, or formal
verification~\cite{yang_vert_2025,bhatia_verified_2024}, to ensure
correctness.  State-of-the-art approaches in code translation translate
whole repositories at a
time~\cite{wang_program_2025,zhang_scalable_2025,ibrahimzada_alphatrans_2025,shiraishi_smartc2rust_2026}.

Despite this progress, prior work in program translation has avoided translating repositories that heavily use external libraries due to two key limitations:

First, LLM performance in code translation degrades when the source code
makes use of external libraries~\cite{gu_effectiveness_2025}.  LLMs
hallucinate non-existent libraries and library APIs in the target language,
and they misuse existing library APIs.

Second, prior techniques often cannot reason about the semantic equivalence
of a source function and its translation when they make use of external
libraries.  Many techniques cannot reason about equivalence of
functions at
all~\cite{xu_optimizing_2025,wang_program_2025,shiraishi_smartc2rust_2026,nitin_c2saferrust_2025,hong_type-migrating_2025,cai_rustmap_2025,zhang_multilingual_2023,pan_lost_2024,dehghan2025translating,wang2025evoc2rust,wang_repotransbench_2024,li_adversarial_2025}
— they rely on end-to-end tests (e.g.  executing a compiled binary like a
command line interface) or hand-written unit tests to validate the
correctness of the translation.  Techniques that \textit{do} attempt to
reason about equivalence of functions make assumptions that often do
not hold when the source function and its translation use external
libraries.  Most techniques assume the input parameter types in the source
function are structurally identical to the input parameter types in the
target
language~\cite{zhang_scalable_2025,yang_vert_2025,shetty_syzygy_2024,zhou_llm-driven_2025,ibrahimzada_alphatrans_2025,eniser_automatically_2024,bhatia_verified_2024,abid_gluetest_2024,ke_advancing_2025}. 
When these input types are defined in external libraries, this assumption
often does not hold.  In addition, all of these works assume that all fields of
complex data types can be accessed directly, which is not possible or
extremely difficult in modern languages that support opaque
types~\cite{zhou_llm-driven_2025}.

\paragraph{Contributions.} In this work, we present, to the best of our knowledge, the first code
translation approach that addresses both of these limitations.  Our approach
is capable of validating the equivalence of pairs of functions with
structurally dissimilar and opaque input parameter types, and uses a RAG
architecture to significantly reduce hallucination of external library APIs. 
We focus the presentation on translating whole Go repositories to Rust due to its practical
relevance~\cite{eniser_automatically_2024} and the fact that Go has a well-defined notion of libraries, but our ideas apply to most
modern languages.  
We evaluated our approach on six real-world Go repositories with non-trivial dependencies, assessing both compilability and differential I/O equivalence. Across all benchmarks, we obtain an average success rate of \(95.83\%\) for both measures. This is an increase of up to 100\% in the most dependency-heavy case and over 2x on average.
%

\section{Motivating Example}
\label{sec:overview}

We illustrate key challenges through an example from the \textbf{age} repository~\cite{age} in Figure~\ref{fig:ed255192curve25519}, showing a Go function and its Rust translation performing cryptographic operations. Our approach relies on an LLM as the primary code generator.

\begin{figure}[htbp]
\begin{minipage}{0.47\textwidth}
\begin{lstlisting}[language=go, basicstyle=\tiny\ttfamily,escapechar=\%, numbers=left, xleftmargin=3em]
import "crypto/sha512"
import "crypto/ed25519"
func ed25519PrivateKeyToCurve25519(
        pk ed25519.PrivateKey) []byte {
    h := sha512.New()
    h.Write(pk.Seed())
    out := h.Sum(nil)
    return out[:curve25519.ScalarSize]
}
%\newline%
\end{lstlisting}
\end{minipage}
\begin{minipage}{0.47\textwidth}
\begin{lstlisting}[language=rust, basicstyle=\tiny\ttfamily, numbers=left, xleftmargin=3em]
use sha2::{Sha512, Digest};
use ed25519_dalek::SigningKey;
pub(crate) fn ed25519_private_key_to_curve25519
        (sk: &SigningKey) -> [u8; 32] {
    let mut hasher = Sha512::new();
    hasher.update(sk.to_bytes());
    let digest = hasher.finalize();
    let mut output = [0u8; 32];
    output.copy_from_slice(&digest[..32]);
    output
}
\end{lstlisting}
\end{minipage}
\caption{Go function from \textbf{age}~\cite{age} and its Rust translation}
\label{fig:ed255192curve25519}
\end{figure}

\paragraph{\textbf{Challenge 1: Cross-language library and \api{} mapping.}}
In Figure~\ref{fig:ed255192curve25519}, on the \golang side, the function 
calls \gocode{sha512.New()} and \gocode{h.Write(pk.Seed())} (Go lines~5--6).
While the intent is straightforward, a translator must still decide
\emph{which} target crates and \api{}s best match the source calls. In particular, a possible Rust translation
constructs a hasher with \rustcode{Sha512::new()} (Rust line~5), feeds bytes derived from the signing key via \rustcode{sk.to\_bytes()} (Rust line~6), and finalises the digest (Rust lines~7--9).
When relying on an LLM’s parametric knowledge, this step is brittle: the model
may select \emph{incorrect} or \emph{non-existent} target crates/\api{}s, leading
to compilation failures (see Section~\ref{sec:evaluation}).

 

\paragraph{Challenge 1.1: Rust import resolution.}
In \rustlang, selecting the right target \api{} is only half the battle: the translation
must also recover the \emph{call-enabling imports} that make the \api{} usable at the call site. We
treat this as a key sub-challenge of Challenge~1.
This is brittle because \rustlang callability is often \emph{non-local}. First, many methods are
\emph{trait-provided} and are only available when the defining trait is in scope. Second, crates
commonly restructure their public surface via \emph{re-exports}, so importable paths often differ from defining modules. As a result, LLMs regularly emit code naming the correct \api{} but failing to compile due to incorrect imports.

Figure~\ref{fig:ed255192curve25519} illustrates both pitfalls. Calling
\rustcode{Sha512::new()} requires importing the \rustcode{Digest} trait
(re-exported as \rustcode{sha2::Digest}), while \rustcode{SigningKey} is defined
in a private module but exposed at the crate root. For this example, the
translator must therefore recover the correct imports
\rustcode{sha2::\{Sha512, Digest\}} and \rustcode{ed25519\_dalek::SigningKey}
(Rust lines~1--2).

\begin{figure}[t]
\centering
\begin{minipage}{0.45\textwidth}
\begin{lstlisting}[language=go, basicstyle=\tiny\ttfamily,escapechar=\%]
import "crypto/rsa"
import "golang.org/x/crypto/ssh"

type RSAIdentity struct {
	k      *rsa.PrivateKey
	sshKey ssh.PublicKey
}
\end{lstlisting}
\caption{\golang type \gocode{RSAIdentity}}
\label{fig:go-rsaidentity}
\end{minipage}
\hfill
\begin{minipage}{0.45\textwidth}
\begin{lstlisting}[language=rust, basicstyle=\tiny\ttfamily,escapechar=\%]
use rsa::RsaPrivateKey;
use russh::keys::PublicKey;

pub struct RsaIdentity {
    pub(crate) k: RsaPrivateKey,
    pub(crate) ssh_key: PublicKey,
}
\end{lstlisting}
\caption{Translated \rustlang type \rustcode{RsaIdentity}}
\label{fig:rust-rsaidentity}
\end{minipage}
\label{fig:rsaidentity-pair}
\end{figure}

\paragraph{\textbf{Challenge 2: Cross-language validation with opaque library types.}} 

Even if a translator resolves the right crates and \api{}s
(Challenge~1) and produces \emph{callable} Rust code with correct imports
(Challenge~1.1), the translation still comes with \emph{no evidence
of equivalence}. Validating repository-level translations requires \emph{cross-language comparison}, yet Go and Rust share neither a common runtime nor value representation.

In principle, one would like to \emph{at least} validate a translated function by executing the
Go source on unit-test inputs, then checking that the Rust translation produces
matching outputs. However, this
seemingly straightforward scheme breaks down in practice because values cannot
be exchanged directly between Go and Rust, and the relevant inputs/outputs
frequently involve \emph{library-defined types} (e.g., \gocode{ed25519.PrivateKey}
and \rustcode{SigningKey}). Such types are often \emph{opaque} (private fields),
lack stable (de)serialization interfaces, and can only be constructed or
inspected through carefully chosen public \api{}s. As a result, validation cannot
rely on field-level conversion or structural matching; it requires synthesising
\emph{behaviour-preserving adapters} using only public library \api{}s.

This difficulty becomes even more pronounced for user-defined wrappers such as
\gocode{RSAIdentity} (Figure~\ref{fig:go-rsaidentity}) and its Rust translation
\rustcode{RsaIdentity} (Figure~\ref{fig:rust-rsaidentity}), which encapsulate
cryptographic key types drawn from different ecosystems (e.g.,
\gocode{rsa.PrivateKey}, Figure~\ref{fig:go-rsa-privatekey}, versus
\rustcode{rsa::RsaPrivateKey}). Here, even constructing a single corresponding
test input requires navigating interface types, trait-scoped methods, and
crate-specific encodings, making cross-language differential testing very challenging.
In particular, synthesising transformations between \gocode{rsa.PrivateKey},
\gocode{ssh.PublicKey}, and their \rustlang counterparts cannot rely on structural,
field-by-field conversion: \gocode{ssh.PublicKey} is an interface, while
\gocode{rsa.PrivateKey} contains nested library types (e.g., \gocode{big.Int} and
\gocode{PrecomputedValues}, Figure~\ref{fig:go-rsa-privatekey}) whose values must
satisfy library-specific invariants. As a result, correct interoperation typically
requires using the libraries' public constructors and (de)serialization routines,
rather than directly accessing or reconstructing internal fields.

\begin{figure}[t]
\centering
\begin{minipage}{0.48\textwidth}
\centering
\begin{lstlisting}[language=go, basicstyle=\tiny\ttfamily, escapechar=\%]
type PrivateKey struct {
    PublicKey            // public part.
    D         *big.Int   // private exponent
    Primes    []*big.Int // prime factors of N,
                         // has >= 2 elements.

    // Precomputed contains precomputed values that
    // speed up RSA operations, if available. It
    // must not be modified.
    Precomputed PrecomputedValues
}
\end{lstlisting}
\caption{Definition of \gocode{rsa.PrivateKey}.}
\label{fig:go-rsa-privatekey}
\end{minipage}
\hfill
\begin{minipage}{0.48\textwidth}
\centering
\begin{tcolorbox}[
    colback=prompt_bg,
    colframe=prompt_title,
    subtitle style={boxrule=0.4pt, colback=yellow!50!blue!25!white, colupper=black},
]
\scriptsize
...\\
- Translation of library \api{}s\\
You are advised to use the following \api{}s from the crate \rustcrate{crate}\\
1. \$\{api\_1.import\_paths\}\$: \$\{api\_1.signature\}\$, \$\{api\_1.docs\}\$\\
2. \$\{api\_2.import\_paths\}\$: \$\{api\_2.signature\}\$, \$\{api\_2.docs\}\$\\
...\\
\end{tcolorbox}
\caption{RAG-enhanced context in the translation prompts.}
\label{fig:prompt-extend-translation}
\end{minipage}
\end{figure}

%% file: sections/overview.tex
\section{High-Level Approach}
\label{sec:approach}

At a high level, our pipeline performs \emph{modular, function-level translation}:
each type definition and function is translated in isolation by an LLM, guided by a compact
\emph{dependency summary} capturing the in-scope types and the signatures of
direct callees. This design is in the spirit of prior modular LLM-based translation approaches such as
\oxidizer~\cite{zhang_scalable_2025}.



\subsection{Stage 1: Cross-Language Library Mapping (Challenge 1)}

As discussed in Section \ref{sec:overview}, LLMs often fail to resolve
\emph{which} \rustlang crate and \api{} best matches a given \golang call. Even
for standard-library usage, the correct counterpart may live in a third-party
Rust crate rather than in \rustcode{std}, and multiple candidate crates may
provide overlapping functionality.

To mitigate this, we augment each translation prompt with explicit \emph{library-resolution context}.
Specifically, for every Go function we (i) extract the set of Go library APIs it invokes, and
(ii) retrieve a small set of semantically relevant \emph{candidate Rust APIs}. 
These candidates are obtained via a retrieval-augmented pipeline briefly described below  (and in more detail in Section~\ref{sec:knowledge-base}): 

\paragraph{Crate matching.}
First, we map each Go package \gopackage{pkg} to a corresponding Rust crate \rustcrate{crate}. We assume an \emph{injective} mapping: each Go package maps to a single Rust crate, and all APIs from \gopackage{pkg} are resolved within \rustcrate{crate}. This keeps retrieval focused and prompt budgets manageable. While this can be relaxed to allow multiple crates per package, we found injectivity to be a useful practical approximation.

We obtain candidate crates by combining (i) LLM-proposed matches based on package-level
documentation summaries and (ii) keyword-based search over \texttt{crates.io}, followed by
documentation-based reranking, with mappings pre-computed before the main translation loop. For the running example in
Figure~\ref{fig:ed255192curve25519}, this yields \rustcrate{sha2} for \gocode{crypto/sha512} and
\rustcrate{ed25519\_dalek} for \gocode{crypto/ed25519}. 

\paragraph{Per-crate knowledge base.}
For each matched target crate \rustcrate{crate}, we precompute a crate-specific knowledge base
$\mathcal{K}_\rustcrate{crate}$ containing its public APIs.
For each library API used in the source function, we then query $\mathcal{K}_\rustcrate{crate}$
using a textual description that combines the Go API name and its documentation. For example,
for \gocode{sha512.New} we construct:
$\text{Query}_{\text{\gocode{sha512.New}}} = \text{\gocode{sha512.New}:}$
``{\it New returns a new hash.Hash computing the SHA-512 checksum.}'',
and retrieve candidates from $\mathcal{K}_\rustcrate{sha2}$. In our system,
\rustcode{Sha512::new} is ranked highest.




\paragraph{Rust import resolution.}
Given that, as explained in Section~\ref{sec:overview}, import resolution is non-trivial in Rust,
we extend our crate-specific knowledge base $\mathcal{K}_\rustcrate{crate}$ to
capture not only \emph{what} APIs exist, but also \emph{how} they can be imported and used.
Concretely, each entry in $\mathcal{K}_\rustcrate{crate}$ stores a publicly usable \api{} together with
its documentation and one or more \emph{valid, public import paths}. 
%
In our running example, the top match for \gocode{sha512.New} is \rustcode{Sha512::new}, returned
with the call-enabling imports \rustcode{sha2::\{Sha512, Digest\}}, enabling the Rust translation
in Figure~\ref{fig:ed255192curve25519}. Likewise, the Go call \gocode{pk.Seed()} (Go line~6) is
resolved to \rustcode{SigningKey::to\_bytes()} (Rust line~6) via
$\mathcal{K}_\rustcrate{ed25519\_dalek}$, using the correct root-level import path
\rustcode{ed25519\_dalek::SigningKey}.

\subsection{Stage 2: Generating the Translation}

\input{sections/translation}

\input{sections/overview-validation}

%% file: sections/translation.tex
As aforementioned, with the library knowledge bases in place, we extend the 
standard \emph{dependency summary}---which captures the
in-scope types and the signatures of direct callees--- in the translation prompt with a small set of relevant
\rustlang \api{} suggestions.
When translating a given \golang item, if the item invokes a \golang \api{}
\gocode{pkg.f} from package \gopackage{pkg}, we query the precomputed knowledge
base \(\mathcal{K}_{\rustcrate{crate}}\) associated with the mapped target crate
\rustcrate{crate}.
The query \(\text{Query}_{\gocode{pkg.f}}\) combines the documentation and type
signature of \gocode{pkg.f}, and we retrieve the top-ranked \api{} candidates
(three in our implementation).
Each suggested \api{} is injected into the prompt together with its
\emph{call-enabling import paths}, type signature, and documentation
(\autoref{fig:prompt-extend-translation}).

%% file: sections/overview-validation.tex
\begin{figure}[t]
\centering
\begin{minipage}{0.45\textwidth}
\begin{lstlisting}[language=proto, basicstyle=\tiny\ttfamily,escapechar=\%]
message RsaIdentity {
  optional bytes k = 1;
  bytes ssh_key    = 2;
}
\end{lstlisting}
\caption{\protobuf schema ($T_{\mathrm{Proto}}$) for types \gocode{RSAIdentity} (Figure~\ref{fig:go-rsaidentity}) and \rustcode{RsaIdentity}
(Figure~\ref{fig:rust-rsaidentity})}
\label{fig:proto-idf}
\end{minipage}
\hfill
\begin{minipage}{0.45\textwidth}
\begin{lstlisting}[language=go, basicstyle=\tiny\ttfamily,escapechar=\%]
type ProtoRSAIdentity struct {
    K      []byte
    SshKey []byte
}
\end{lstlisting}
\caption{Flattened \golang carrier type produced by the \protobuf compiler ($T_{\mathrm{Proto}}^{\go\golang}$)}
\label{fig:flat-go-type}
\end{minipage}
\end{figure}

\begin{figure}[t]
\centering
\begin{minipage}{\textwidth}
\begin{lstlisting}[language=go, basicstyle=\tiny\ttfamily,escapechar=\%]
func ToProtoRSAIdentity(input *RSAIdentity) (ProtoRSAIdentity, error) {
	if input == nil {
		return ProtoRSAIdentity{}, errors.New("input RSAIdentity is nil")
	}

	proto := ProtoRSAIdentity{}

	// Convert RSA private key to PKCS#1 DER format
	if input.k != nil {
		proto.K = x509.MarshalPKCS1PrivateKey(input.k)
	}

	// Convert SSH public key to wire format
	if input.sshKey != nil {
		proto.SshKey = input.sshKey.Marshal()
	}

	return proto, nil
}
\end{lstlisting}
\caption{From original \golang type to flattened \golang carrier type}
\label{fig:toproto-go}
\end{minipage}
\end{figure}

\begin{figure}[t]
    \centering
\begin{minipage}{\textwidth}
\begin{lstlisting}[language=rust, basicstyle=\tiny\ttfamily,escapechar=\%]
use russh::keys::PublicKey;
use rsa::pkcs1::DecodeRsaPrivateKey;
fn from_proto_rsa_identity(proto: &ProtoRSAIdentity) -> Result<RSAIdentity> {
    let k = proto.k.as_ref()
        .ok_or_else(|| anyhow::anyhow!("Missing RSA private key in ProtoRSAIdentity"))?;
    let k = RsaPrivateKey::from_pkcs1_der(k)
        .context("Failed to parse RSA private key from PKCS#1 DER")?;

    let ssh_key = PublicKey::from_bytes(&proto.ssh_key)
        .context("Failed to parse SSH public key")?;

    Ok(RSAIdentity {
        k,
        ssh_key,
    })
}
\end{lstlisting}
\caption{From original \rustlang type to flattened \rustlang carrier type}
\label{fig:fromproto-rust}
\end{minipage}
\end{figure}

\subsection{Cross-Language Validation with Opaque Library Types (Challenge 2)}
\label{sec:validation}


To enable end-to-end validation, we adopt a \textbf{\emph{two-stage}} pipeline.
First, we establish \textbf{\emph{cross-language interoperability}} between the source and target languages 
by synthesising adapters that allow values to cross the Go--Rust boundary.
Second, once cross-language execution becomes possible, we validate
\textbf{\emph{I/O equivalence}} between each Go function and its Rust translation
via differential testing.  


\paragraph{Cross-language interoperability via a schema-backed carrier representation.}
Our key idea is to bridge Go and Rust via a \emph{shared, language-neutral} intermediate data
format (IDF) specified as a single \protobuf schema $T_{\mathrm{Proto}}$~\cite{protobuf}. We
synthesise this schema with an LLM from the source Go type definition: for example, given
\gocode{RSAIdentity} (Figure~\ref{fig:go-rsaidentity}), the model generates the IDF schema in
Figure~\ref{fig:proto-idf}.

From $T_{\mathrm{Proto}}$, the \protobuf compiler generates \emph{language-specific carrier types}
$T_{\mathrm{Proto}}^{\go\golang}$ and $T_{\mathrm{Proto}}^{\rust\rustlang}$, together with
(de)serialization code that maps these carriers to and from the shared wire format. Concretely,
$T_{\mathrm{Proto}}^{\go\golang}$ is a Go struct type (e.g., \gocode{ProtoRSAIdentity} in
Figure~\ref{fig:flat-go-type}), while $T_{\mathrm{Proto}}^{\rust\rustlang}$ is the corresponding Rust
type generated from the same schema. These carrier types are isomorphic by construction
and provide a common representation on which we can build cross-language adapters.

%


This reduces cross-language interoperability to synthesising two pairs of conversions around the
shared carrier representation. On the source side, we require adapters that map between the Go
types used by the original program and the Go carrier type $T_{\mathrm{Proto}}^{\go\golang}$; on the
target side, we analogously require adapters between the Rust carrier type
$T_{\mathrm{Proto}}^{\rust\rustlang}$ and the Rust types used by the translation. In both cases, the
carrier acts as the common ``meeting point'' through which values can be transported across the
Go--Rust boundary via \protobuf serialisation.

Crucially, these conversions cannot be implemented via field-level access. Realistic code frequently
wraps library-defined objects whose internal representations are opaque or constrained by
invariants enforced by the library. For example, \gocode{RSAIdentity}
(Figure~\ref{fig:go-rsaidentity}) and its translation \rustcode{RsaIdentity}
(Figure~\ref{fig:rust-rsaidentity}) encapsulate cryptographic key types such as
\gocode{ssh.PublicKey} and \rustcode{russh::keys::PublicKey}, which can only be constructed and
observed through restricted public APIs rather than by directly inspecting or populating fields.

\begin{figure}[t]
\centering
\begin{minipage}{0.45\textwidth}
\begin{lstlisting}[language=go, basicstyle=\tiny\ttfamily,escapechar=\%]
import "crypto/rsa"

type RSAKeyPair struct {
    // user-defined wrapper field
    identity RSAIdentity   
    // library-defined type field
    pubKey *rsa.PublicKey  
}
\end{lstlisting}
\caption{Composite \golang type \gocode{RSAKeyPair}}
\label{fig:go-rsakeypair}
\end{minipage}
\hfill
\begin{minipage}{0.45\textwidth}
\begin{lstlisting}[language=rust, basicstyle=\tiny\ttfamily,escapechar=\%]
use rsa::RsaPublicKey;

pub struct RsaKeyPair {
    pub(crate) identity: RSAIdentity,
    pub(crate) pub_key: RsaPublicKey,
}
\end{lstlisting}
\caption{Translated \rustlang type \rustcode{RsaKeyPair}}
\label{fig:rust-rsakeypair}
\end{minipage}
\label{fig:rsarecipient-pair}
\end{figure}
\begin{figure}[t]
\centering
\begin{minipage}{\textwidth}
\begin{lstlisting}[language=rust, basicstyle=\tiny\ttfamily,escapechar=\%, numbers=left, xleftmargin=4em]
fn from_proto_rsa_key_pair(proto: &ProtoRSAKeyPair) -> Result<RSAKeyPair> {
    let identity = proto.identity.as_ref()
        .ok_or_else(|| anyhow::anyhow!("Missing RSA identity in ProtoRSAKeyPair"))?;
    let identity = from_proto_rsa_identity(identity)
        .context("Failed to parse RSA identity")?;

    let pub_key_der = proto.pub_key.as_ref()
        .ok_or_else(|| anyhow::anyhow!("Missing RSA public key in ProtoRSAKeyPair"))?;
    let pub_key = RsaPublicKey::from_pkcs1_der(pub_key_der)
        .context("Failed to parse RSA public key from PKCS#1 DER")?;

    Ok(RSAKeyPair {
        identity,
        pub_key,
    })
}
\end{lstlisting}
\caption{From flattened \protobuf type to \rustlang type \rustcode{RSAKeyPair}}
\label{fig:rust-from-proto-rsakeypair}

\end{minipage}
\label{fig:rsarecipient-pair}
\end{figure}

\paragraph{Synthesising adapters using public APIs.}
Given the carrier types generated from $T_{\mathrm{Proto}}$, we synthesise the required
high-level adapters by querying our previously constructed library knowledge bases for relevant \emph{public} \api{}s and constraining the LLM to assemble conversions exclusively from these APIs. This is essential for
opaque library types, where direct structural conversion is impossible and correct construction must
go through library-provided routines. For example, the Go-to-carrier adapter
(Figure~\ref{fig:toproto-go}) serialises keys using \gocode{x509.MarshalPKCS1PrivateKey} and the
interface method \gocode{ssh.PublicKey.Marshal}. On the Rust side, the carrier-to-Rust adapter
(Figure~\ref{fig:fromproto-rust}) reconstructs keys via crate-specific decoding routines (e.g.,
PKCS\#1 decoding and public-key parsing).

\paragraph{Handling complex types.}
In practice, values passed across the boundary can be arbitrarily structured, combining multiple
components---including nested user-defined wrappers and opaque library-defined objects. For
instance, \gocode{RSAKeyPair} and its Rust translation \rustcode{RsaKeyPair}
(Figures~\ref{fig:go-rsakeypair} and~\ref{fig:rust-rsakeypair}) embed both wrapper-level fields and
library types. To scale to such cases, we generate carrier schemas and adapters modularly, reusing
previously validated conversions for each component. For example, in
Figure~\ref{fig:rust-from-proto-rsakeypair} (line~4), the carrier-to-Rust adapter for
\rustcode{RsaKeyPair} reuses the carrier-level conversion for \rustcode{RsaIdentity} from
Figure~\ref{fig:fromproto-rust}.



\paragraph{Schema and carriers validation.}
Since both the carrier schema (Figure~\ref{fig:proto-idf}) and the adapters
(Figures~\ref{fig:toproto-go} and~\ref{fig:fromproto-rust}) are LLM-synthesised, we validate them
via \emph{round-trip} checks before using them for I/O equivalence testing. Concretely, on the Go
side we require the forward/backward adapters $f_{\mathrm{Go}\rightarrow \mathrm{Proto}}$ and
$f_{\mathrm{Proto}\rightarrow \mathrm{Go}}$ (between the original \golang types and $T_\mathrm{Proto}$) to satisfy $f_{\mathrm{Proto}\rightarrow \mathrm{Go}}(f_{\mathrm{Go}\rightarrow \mathrm{Proto}}(v))=v$ on representative
values $v$ (e.g., obtained from unit-test executions), providing evidence that the carrier preserves
enough information to reconstruct the original Go value. Once Go $\leftrightarrow$ Proto roundtripping succeeds, we validate full cross-language interoperability by composing the one-sided adapters into end-to-end transformations between the source and target types: \(f_{\mathrm{Rust}\rightarrow \mathrm{Go}}(f_{\mathrm{Go}\rightarrow \mathrm{Rust}}(v))=v\). Passing these checks is a necessary precondition for transporting inputs and outputs across the language boundary, after which we can meaningfully test end-to-end I/O equivalence. 

On round-trip checks failure, we regenerate adapters while keeping the translation fixed. If retries exceed a predefined budget, we regenerate the translation and revalidate.



\paragraph{Differential I/O validation.}
Once type interoperability is established for a function’s input and output types, we reuse the
validated adapters to execute the Go source and Rust translation on the same unit-test inputs (mapped through the
carrier representation), and compare their observable outputs. This enables
repository-level differential testing even when translated code manipulates
library-defined types that would otherwise be inexpressible across the language
boundary.

If I/O validation fails, we freeze the translated types, function signatures, and
adapters, and regenerate only the function body---up to a fixed budget---until the
validation succeeds or the budget is exhausted.

%% file: sections/knowledge-base.tex
\section{Crate-Scoped Knowledge Base Construction}
\label{sec:knowledge-base}

Our RAG pipeline relies on \emph{crate-scoped} retrieval. Given a target Rust crate, we construct a
queryable knowledge base that captures its publicly usable APIs together with their \emph{valid
import paths}, enabling precise API retrieval during translation. 
This section describes the
construction of this knowledge base.




\subsection{Indexing \api{}s for a Crate}
\label{sec:per-crate-kb}

For each target crate \rustcrate{crate} in the pre-computed
library mapping,
we
construct a crate-specific knowledge base $\mathcal{K}_\rustcrate{crate}$ to support retrieving
relevant Rust APIs from natural-language queries derived from Go documentation and signatures.

Formally, for each crate we first construct an inventory
mapping every API item to its documentation and a set of call-enabling import paths:
\begin{eqnarray*}
\itemsmap &:& \text{API} \rightarrow \text{Doc} \times \mathcal{P}(\text{ImportPath}),
\end{eqnarray*}
We then
expose this inventory through a queryable knowledge base $\mathcal{K}$ that maps a
textual query to a ranked list of API entries\footnote{Here $\mathcal{S}(\cdot)$ denotes
a total ranking over the inventory entries.}:
$\mathcal{K} : \text{String} \rightarrow \mathcal{S}(\itemsmap)$.

\paragraph{\textbf{Structural representation of the \rustlang library documentation.}}
Our $\itemsmap$ construction algorithm operates over a lightweight structural
representation of the Rust crate documentation, inspired by the \textit{rustdoc-json}
format and formalised in~\autoref{fig:rustdoc}. The goal of this representation
is to expose, in a uniform way, the information needed to enumerate a crate's
public API surface and derive the set of valid import paths that may refer to
each API item.

As shown in~\autoref{fig:rustdoc}, a Rust crate consists of a single root module
\(\elem{M}\). Each module stores: (i) module-level documentation (\text{Doc});
(ii) a list of items \(\overline{\elem{I}}\); and (iii) a list of
re-export declarations of the form \rustcode{pub use}~\(\elem{P}\). Re-exports
are essential for $\itemsmap$ because they expose items under additional public
paths, which must be considered when mapping names in natural-language
documentation to concrete API references.

Each item \(\elem{I}\) is either a submodule, a type \(\elem{t}\), or a trait
\(\mathbb{T}\), and every item carries its own documentation. Traits
(\(\mathbb{T}\)) define a set of method signatures \(\overline{\elem{F}}\),
while types (\(\elem{T}\)) are associated with a set of \rustcode{impl} blocks
(\(\overline{\text{ImplBlock}}\)) that provide methods, including trait
implementations. In~\autoref{fig:rustdoc} we represent an \text{ImplBlock} as a
trait implementation \(\rustcode{impl}\ \mathbb{T}\ \rustcode{for}\ \elem{T}\)
together with its method set \(\overline{\elem{F}}\); although the definition is
simplified, our implementation handles all Rust \rustcode{impl} forms in
practice. Finally, paths \(\elem{P}\) encode fully-qualified names of the form
\(\elem{M}(::\elem{N})^*::\elem{I}\), which serve as the canonical import paths
used by $\itemsmap$.


\begin{figure}
\centering
\begin{minipage}{0.48\linewidth}
\[
\begin{array}{llcl}
  \text{Module} & \elem{M} & :=   & \text{Doc} \times \overline{\elem{I}} \times \overline{\text{\rustcode{pub use} } \elem{P}} \\
  \text{Item} & \elem{I} & := & \elem{T}\ |\ \mathbb{T}\ |\ \elem{f} \\
  \text{Type} & \elem{T} & := & \text{Doc} \times \overline{\text{ImplBlock}}\\
  \text{Trait} & \mathbb{T} & := & \text{Doc} \times \overline{\elem{f}}
\end{array}
\]
\end{minipage}
\hfill
\begin{minipage}{0.48\linewidth}
\[
\begin{array}{llcl}
  \text{ImplBlock} & & := & \rustcode{impl}\ \mathbb{T}\ \rustcode{for}\ \elem{T} \times \overline{\elem{f}}\\
  \text{Function} & \elem{f} & := & \text{Doc}\\
  \text{Path} & \elem{P} & := & \elem{M}(::\elem{M})^*::\elem{I}
\end{array}
\]
\end{minipage}

\caption{Structural representation of Rust library documentation.}
\label{fig:rustdoc}
\end{figure}

\paragraph{\rustlang $\itemsmap$ construction.}
We make use of the structural representation of the \rustlang library documentation
to build $\itemsmap$ for a Rust crate. For this purpose, we systematically enumerate its publicly usable API surface
together with \emph{call-enabling} import paths. This requires accounting for deep module
hierarchies, \rustcode{pub use} re-exports (including re-exports from transitive dependencies), and
trait-scoped method callability. Our algorithm therefore traverses the public module tree, resolves
both internal and external re-exports using precomputed inventories for dependencies, and extracts
types, traits, and methods. For trait methods, we initially record placeholder trait references and
later rewrite them to concrete import paths in a final revision pass. We show that all generated
paths are valid under Rust visibility rules. Full details, including algorithms and proofs, appear in
Appendix~\ref{app:itemsmap-rust}.

\paragraph{\golang $\itemsmap$ construction.}

Notably, when synthesising the glue code between \golang and \rustlang (see
Section~\ref{sec:interop}), we additionally require knowledge bases for
\golang libraries.
Constructing $\itemsmap$ for \golang is straightforward given its simple type and package system---referencing an \api{} requires only the library name and \api{} name (e.g., \gocode{ssh.PublicKey} in~\autoref{fig:go-rsaidentity}). This can be achieved via webpage scraping or similar approaches.  In this work we parse the
output of the \gocode{go doc} command, which provides comprehensive library documentation.

\subsection{Constructing the Knowledge Base}
\label{sec:map-api}

From the per-crate inventory $\itemsmap$, we build a knowledge base $\mathcal{K}$ for retrieving \emph{relevant} \api{} candidates from natural-language queries. Concretely, $\mathcal{K}$ ranks entries by relevance to a query $Q$ (e.g., derived from Go documentation and signatures) and returns top matches with their call-enabling import paths.

\paragraph{Two-stage retrieval and reranking.}
Directly applying an LLM reranker over the full \itemsmap{} is infeasible for
real-world crates, which may expose thousands of \api{}s. Instead, we adopt a
standard two-stage information retrieval pipeline: an efficient first-stage
retriever followed by a more accurate reranker.

Given a query \(Q\), we first apply a dense bi-encoder retriever in the style of
Dense Passage Retrieval (DPR)~\cite{karpukhin_dense_2020} to score \api{} entries by
semantic similarity between \(Q\) and their documentation, and select a small
candidate set (top-\(k\), with \(k=35\) in our implementation beyond which we observed diminishing returns in recall). We then rerank these candidates using a stronger cross-encoder model in the style of
BERT-based reranking~\cite{nogueira_passage_2019}. Finally, we apply an LLM-based
list-wise reranker to obtain the final ordering~\cite{pradeep_rankzephyr_2023}.

%% file: sections/validate.tex
\section{Validating the Translation}
\label{sec:translation-validation}

In this section, we develop a validation framework for checking equivalence of function translations in the presence of library dependencies. 
As discussed in Section~\ref{sec:overview}, validating Go-to-Rust translations remains challenging even when the generated code compiles, because the two languages have incompatible value representations and pervasive library-defined types. To support end-to-end validation, we employ a \textbf{\emph{two-stage}} pipeline: we first establish \textbf{\emph{cross-language interoperability}} 
by synthesising adapters that safely serialize values across the Go--Rust boundary, and then leverage these validated adapters to check \textbf{\emph{I/O equivalence}} between each Go function and its Rust translation via differential testing on unit-test derived inputs.

\subsection{Cross-Language Interoperability}
\label{sec:interop}

Our goal is to enable cross-language execution for validation: given a Go value of a type
$T_{\go\golang}$, we must construct a
corresponding Rust value of the translated type $T_{\rust\rustlang}$, and conversely recover a
Go value back. We model this as synthesising a pair of \emph{cross-language adapters}
\[
    \mathcal{A}_{T} : T_{\go\golang} \rightarrow T_{\rust\rustlang} \qquad
    \mathcal{A}^{-}_{T} : T_{\rust\rustlang} \rightarrow T_{\go\golang},
\]
where $\mathcal{A}_{T}$ lifts Go unit-test inputs/outputs into Rust for downstream differential
testing, and $\mathcal{A}^{-}_{T}$ supports round-trip validation. In particular, the composition
$\mathcal{A}^{-}_{T} \circ \mathcal{A}_{T}$ should preserve observable Go values, providing evidence
that (i) the translated Rust type does not discard semantically relevant information and
(ii) the adapters are behaviour-preserving on realistic inputs.

Directly generating $\mathcal{A}_{T}$ and $\mathcal{A}^{-}_{T}$ is difficult because Go and Rust values
have incompatible representations, and library-defined types are often opaque (private fields)
and only manipulable via public APIs. To address this, we
introduce a shared \emph{intermediate representation} based on \protobuf, $T_{\mathrm{Proto}}$, which
serves as a language-neutral carrier for values crossing the Go--Rust boundary.



\paragraph{\textbf{\protobuf schema generation.}}
We next describe how the carrier schema $T_{\mathrm{Proto}}$ is synthesised from
a composite source Go type $T_{\go\golang}$

A key distinction is between \emph{language- or library-defined} composite types and
 \emph{user-defined} ones. With respect to the former, given a composite type \(T_{\go\golang}\), for each one of its fields the LLM selects an appropriate \protobuf
encoding. Primitive Go types (e.g., \gocode{int32}, \gocode{string}, \gocode{bool}) are mapped directly to their corresponding \protobuf scalar types. In contrast, complex or opaque fields---particularly library-defined types with private internal representations---are conservatively represented as byte arrays. This byte-array fallback ensures that values remain transmissible across the language boundary even when their internal structure is inaccessible or lacks a natural \protobuf analogue.


For user-defined composite types with nested user-defined dependencies, we
generate schemas compositionally in dependency order.
The LLM 
incrementally constructs
schemas for types that reference previously processed ones.
When generating a schema for a type \(T\) that refers to another user-defined
type \(T'\), we instruct the LLM to reuse the existing \protobuf definition of
\(T'\), ensuring consistency across schemas and avoiding redundant
declarations. 

This structural exposure of user-defined composites is crucial for later stages: it enables the synthesis of glue code that composes field-level conversions and systematically reuses
existing library APIs when mapping between \golang, \rustlang, and the carrier
representation.
By contrast, library-defined types---whose internal structure is often opaque to
clients---are treated conservatively and handled via API-based adapters rather
than field-level access.




\paragraph{\textbf{High-level glue code generation.}}
From the \protobuf schema, the \protobuf compiler generates \emph{language-specific carrier types}
$T_{\mathrm{Proto}}^{\go\golang}$ and $T_{\mathrm{Proto}}^{\rust\rustlang}$, together with low-level serialization and deserialization routines that encode values into (and decode values from) the shared wire format.
These compiler-generated types and routines provide a lightweight interoperability layer, avoiding explicit FFI manipulation. 
At a conceptual level, the schema \(T_{\mathrm{Proto}}\) defines a language-neutral carrier that is realised concretely in both \golang and \rustlang.
We therefore treat \(T_{\mathrm{Proto}}\) as a type that \emph{exists in both
languages}, and reduce the interoperability problem to synthesising
\emph{high-level glue code} that connects \(T_{\mathrm{Proto}}\) to the source
type \(T_{\go\golang}\) and the translated type \(T_{\rust\rustlang}\).

With $T_{\mathrm{Proto}}$ fixed, cross-language interoperability reduces to generating \emph{language-local} glue code converting between each $T_{\mathrm{Lang}}$ and $T_{\mathrm{Proto}}$:
\[
    \mathcal{A}_{T_{\mathrm{Lang}}}  :  T_{\mathrm{Lang}} \rightarrow T_{\mathrm{Proto}} \qquad
    \mathcal{A}^{-}_{T_{\mathrm{Lang}}}  :  T_{\mathrm{Proto}} \rightarrow T_{\mathrm{Lang}},
\]
where $\mathrm{Lang} \in \{\golang,\ \rustlang\}$. These adapters bridge rich, language-specific data structures and the shared \protobuf encoding.

By construction, $T_{\mathrm{Proto}}$ is lightweight, containing only \protobuf-native
primitives like scalars, enumerations, byte arrays, and shallow aggregates.
In contrast, $T_{\mathrm{Lang}}$ types may be deeply nested, include user-defined structures, and depend on library types accessible only through public APIs.

We synthesise the glue functions using an LLM, prompted with the definitions of
$T_{\mathrm{Lang}}$ and $T_{\mathrm{Proto}}$. For composite types, we instruct the model to generate glue
code compositionally, reusing existing conversions whenever possible and relying on retrieved
library guidance when necessary:
\begin{enumerate}
    \item \textbf{Reuse of existing adapters.} If $T_{\mathrm{Proto}}$ (or $T_{\mathrm{Lang}}$) depends on a
    nested type $U$, we reuse previously generated adapters for $U$ to avoid duplicating conversion
    logic and to maintain consistency across the repository.

    \item \textbf{Knowledge-guided library conversions.} If $T_{\mathrm{Lang}}$ contains a type from library \textit{lib}, we augment the prompt with APIs retrieved from the corresponding
    knowledge base $\mathcal{K}_{\textit{lib}}(\text{Query}_{T_{\mathrm{Lang}}\leftrightarrow T_{\mathrm{Proto}}})$.
    The query $\text{Query}_{T_{\mathrm{Lang}}\leftrightarrow T_{\mathrm{Proto}}}$ includes the definitions of
    $T_{\mathrm{Lang}}$ and $T_{\mathrm{Proto}}$ to steer the model towards call-enabled constructors,
    accessors, and (de)serialization routines.
\end{enumerate}

Despite this guidance, LLM-generated glue code remains fallible: the model may select
inappropriate APIs, introduce type mismatches, or produce code that compiles but violates semantics at runtime. We therefore validate adapters via round-trip checking before downstream equivalence testing, as described next.

\paragraph{\textbf{Round-tripping validation.}}
Since the carrier schema $T_{\mathrm{Proto}}$, the translated Rust type $T_{\rust\rustlang}$, and the
glue functions $\mathcal{A}_{T_{\mathrm{Lang}}}$ and $\mathcal{A}^{-}_{T_{\mathrm{Lang}}}$ are all synthesised by an
LLM, we explicitly validate that they preserve information before using them for downstream
equivalence checking. Taking the source type $T_{\go\golang}$ as ground truth, we perform
validation in two phases: we first validate Go $\leftrightarrow$ Proto conversions, and only then
validate full Go $\leftrightarrow$ Rust interoperability.

\paragraph{\golang-side roundtripping.}
In the first phase, we validate that Go-side glue code roundtrips values through the carrier representation without loss:
\begin{definition}[\golang-side round-trip property]\label{def:go-side-roundtrip}
Given a \golang type $T_{\go\golang}$ and transformations
$\mathcal{A}_{T_{\go\golang}} : T_{\go\golang} \rightarrow T_{\mathrm{Proto}}$ and
$\mathcal{A}^{-}_{T_{\go\golang}} : T_{\mathrm{Proto}} \rightarrow T_{\go\golang}$,
the following property holds:
\[
\forall v.\quad \mathcal{A}^{-}_{T_{\go\golang}}\!\bigl(\mathcal{A}_{T_{\go\golang}}(v)\bigr)=v.
\]
\end{definition}
Satisfying Definition~\ref{def:go-side-roundtrip} implies that the forward embedding
$\mathcal{A}_{T_{\go\golang}}$ is injective on the tested values, providing evidence that the
generated schema $T_{\mathrm{Proto}}$ captures all semantically relevant information present in
$T_{\go\golang}$ (i.e., no loss or unintended normalisation is introduced by the carrier).

Establishing this universally quantified property is infeasible in general, 
so we approximate it using concrete values $\overline{v}$ extracted from existing unit tests. A failure indicates
either (i) an inadequate carrier schema or (ii) incorrect Go-side glue code, and triggers regeneration
of \emph{both} $T_{\mathrm{Proto}}$ and the Go $\leftrightarrow$ Proto adapters.

\paragraph{Full roundtripping.}
Once Go $\leftrightarrow$ Proto roundtripping succeeds, we validate full cross-language
interoperability by composing the one-sided adapters into end-to-end transformations between the
source and target types:
\[
\mathcal{A}_{T} = \mathcal{A}^{-}_{T_{\rust\rustlang}} \circ \mathcal{A}_{T_{\go\golang}}
\qquad\text{and}\qquad
\mathcal{A}^{-}_{T} = \mathcal{A}^{-}_{T_{\go\golang}} \circ \mathcal{A}_{T_{\rust\rustlang}}.
\]
We then validate the corresponding round-trip property:
\begin{definition}[Full round-trip property]\label{def:full-roundtrip}
Given a \golang type $T_{\go\golang}$ and transformations
$\mathcal{A}_{T} : T_{\go\golang} \rightarrow T_{\rust\rustlang}$ and
$\mathcal{A}^{-}_{T} : T_{\rust\rustlang} \rightarrow T_{\go\golang}$, the following property holds:
\[
\forall v.\quad \mathcal{A}^{-}_{T}\!\bigl(\mathcal{A}_{T}(v)\bigr)=v.
\]
\end{definition}
As before, we approximate Definition~\ref{def:full-roundtrip} by testing it on a finite set of values. Since the first phase has already established
the correctness of Go $\leftrightarrow$ Proto conversions, failures in this phase can be attributed to
the Rust-side artefacts (the translated type $T_{\rust\rustlang}$ or the Rust $\leftrightarrow$ Proto glue).
Accordingly, when the full round-trip check fails, we regenerate only the Rust-side adapters,
avoiding unnecessary recomputation of $T_{\mathrm{Proto}}$ and the validated Go-side glue code.

\subsection{Validation of I/O Equivalence}

With interoperability established (Section~\ref{sec:interop}), we validate \emph{observable behaviour} via differential testing. Given a Go function $f$ and its LLM-generated Rust translation $g$:
\[
    f: I_{\go\golang} \rightarrow O_{\go\golang} \qquad g: I_{\rust\rustlang} \rightarrow O_{\rust\rustlang}.
\]
For simplicity, we assume single input and output types throughout this section.
Since inputs and outputs may involve
opaque library-defined types, we validate I/O equivalence by comparing executions via the adapters from
Section~\ref{sec:interop}, which transport values across the Go--Rust boundary:

\begin{definition}[I/O Equivalence]\label{def:io-equiv}
Given a \golang function $f$ and the \rustlang translation $g$, as well as 
the transformations generated in the previous step,
\[
    \mathcal{A}_I :  I_{\go\golang} \rightarrow I_{\rust\rustlang} \qquad
    \mathcal{A}_O :  O_{\go\golang} \rightarrow O_{\rust\rustlang}
\]
then $f$ and $g$ are I/O equivalent if $\forall v.\mathcal{A}_O\bigl(f(v)\bigr) = g\bigl(\mathcal{A}_I(v)\bigr)$.
\end{definition}



\paragraph{I/O equivalence on observed inputs.}
We assume the source repository includes unit tests, whose executions yield concrete I/O examples
$\overline{(i,o)}$ for $f$. We treat these examples as a behavioural specification and validate I/O
equivalence by checking that $g$ matches $f$ on the observed inputs. Concretely, letting
$\overline{i}$ denote the set of Go inputs extracted from unit-test runs, we report I/O equivalence
for $(f,g)$ when Definition~\ref{def:io-equiv} holds for every $v \in \overline{i}$ (for $\overline{i}$ the inputs in $\overline{(i,o)}$).

\paragraph{Interpreting the guarantee.}
Our checks do not establish full semantic equivalence between $f$ and $g$; rather, they provide evidence that translation preserves behaviour on observed inputs. To interpret this at the Go-level, we apply the inverse adapter $\mathcal{A}^{-}_{O}$ to the Rust output, yielding:

\begin{lemma}[Go-level behavioural agreement on observed inputs]\label{prop:io-equiv-observed}
Assume the full round-trip property (Definition~\ref{def:full-roundtrip}) holds for the output type $O$ on all values in
$\{f(v)\mid v\in\overline{i}\}$, and assume I/O equivalence holds for all $v\in\overline{i}$.
Then for every $v\in\overline{i}$,
\[
f(v) \;=\; \mathcal{A}^{-}_{O}\!\left(g\left(\mathcal{A}_I(v)\right)\right),
\]
i.e., on the observed inputs, $f$ produces the same outputs as the function
$\mathcal{A}^{-}_{O}\circ g\circ \mathcal{A}_I$.
\end{lemma}

\begin{proof}
Fix any $v \in \overline{i}$. By I/O equivalence (Definition~\ref{def:io-equiv}) we have
$\mathcal{A}_{O}\bigl(f(v)\bigr) \;=\; g\bigl(\mathcal{A}_I(v)\bigr)$.
By the full round-trip property for outputs (Definition~\ref{def:full-roundtrip}), $\mathcal{A}_O$ and
$\mathcal{A}^{-}_{O}$ are defined on $f(v)$ and satisfy
$\mathcal{A}^{-}_{O}(\mathcal{A}_{O}(f(v))) = f(v)$. Applying $\mathcal{A}^{-}_{O}$ to both sides of the
above equation yields $f(v) \;=\; \mathcal{A}^{-}_{O}\!\left(g\left(\mathcal{A}_I(v)\right)\right)$.
Since $v$ was arbitrary, the claim holds for all $v \in \overline{i}$.
\end{proof}

Proposition~\ref{prop:io-equiv-observed} implies the Rust implementation can replace the Go function on observed inputs, even when they involve opaque library types.

%% file: sections/experiment.tex
\section{Experimental Evaluation}
\label{sec:evaluation}


\begin{table}[t]
\centering
\small
\setlength{\tabcolsep}{1pt}
\renewcommand{\arraystretch}{1.1}
\caption{Benchmark description. LoC excludes comments and blank lines (\gocloc~\cite{gocloc} \texttt{code} count). The penultimate column reports the total LoC in external functions invoked by the project, counted once per function. The final column shows the number of distinct library APIs used in the project.}
\resizebox{\columnwidth}{!}{%
\begin{tabular}{r r r r r r}
\toprule
\textbf{Benchmark} & {\makecell{\textbf{Dependency}\\\textbf{Domain}}} & \textbf{Stars/Forks} &
{\makecell{\textbf{\go{Go} LOC}\\\textbf{(no lib)}}} &
{\makecell{\textbf{\go{Go} LOC}\\\textbf{(+ lib)}}} &
{\makecell{\textbf{External}\\ \textbf{Library}\\ \textbf{Items}}} \\
\midrule
\pbkdft~\cite{pbkdft} 
  & Key Derivation 
  & \pbkdftStars/\pbkdftForks 
  & \pbkdftLoC & 2996 & 29 \\

\easyScrypt~\cite{easyscrypt}
  & Crypto 
  & \easyScryptStars/\easyScryptForks 
  & \easyScryptLoC & 5484 & 29 \\

\simpleScrypt~\cite{simplescrypt} 
  & Crypto 
  & \simpleScryptStars/\simpleScryptForks 
  & \simpleScryptLoC & 10517 & 10 \\

\argonTwoHashing~\cite{argon2hashing} 
  & Password 
  & \argonTwoHashingStars/\argonTwoHashingForks 
  & \argonTwoHashingLoC & 6485 & 9 \\

\duration~\cite{duration} 
  & Time Parsing 
  & \durationStars/\durationForks 
  & \durationLoC & 12041 & 9 \\
  
\age~\cite{age} 
  & Crypto 
  & \ageStars/\ageForks 
  & \ageLoC & 4103 & 38 \\

\bottomrule
\end{tabular}
}
\label{tab:benchmarks}
\end{table}

\begin{table}[t]
\centering
\small
\setlength{\tabcolsep}{1.3pt}
\renewcommand{\arraystretch}{1.1}
\caption{Translation results. RAG load is the total number of retrievable API entries across all Rust crates involved in a translation task, computed as the sum of the sizes of the per-crate knowledge bases constructed, and  approximates the size of the RAG candidate search space.}
\resizebox{\columnwidth}{!}{%
\begin{tabular}{c c cc cc cc}
\toprule
& & \multicolumn{2}{c}{\textbf{Full}} & \multicolumn{2}{c}{\textbf{No RAG}} & \multicolumn{2}{c}{\textbf{No imports}} \\
\cmidrule(lr){3-4}\cmidrule(lr){5-6}\cmidrule(lr){7-8}
 & {\makecell{\textbf{RAG} }} 
& {\textbf{Comp.} } & {\textbf{Equiv.}}
& {\textbf{Comp.}} & { \textbf{Equiv.}}
& {\textbf{Comp.} } & { \textbf{Equiv.}} \\
\textbf{Benchmark} & {\makecell{ \textbf{load}}} 
& {\textbf{full/dep}} & { \textbf{full/dep}}
& {\textbf{full/dep}} & { \textbf{full/dep}}
& {\textbf{full/dep}} & { \textbf{full/dep}} \\
\midrule

\pbkdft~\cite{pbkdft}
  & \pbkdftRAGLoad
  & \pbkdftCompPercent/\pbkdftCompPercentLibOnly & \pbkdftIOPercent/\pbkdftIOPercentLibOnly
  & \pbkdftNoRAGCompPercent/\pbkdftNoRAGCompPercentLibOnly & \pbkdftNoRAGIOPercent/\pbkdftNoRAGIOPercentLibOnly
  & \pbkdftRAGNoImportCompPercent/\pbkdftRAGNoImportCompPercentLibOnly & \pbkdftRAGNoImportIOPercent/\pbkdftRAGNoImportIOPercentLibOnly \\

\easyScrypt~\cite{easyscrypt}
  & \easyScryptRAGLoad
  & \easyScryptCompPercent/\easyScryptCompPercentLibOnly
  & \easyScryptIOPercent/\easyScryptIOPercentLibOnly
  & \easyScryptNoRAGCompPercent/\easyScryptNoRAGCompPercentLibOnly
  & \easyScryptNoRAGIOPercent/\easyScryptNoRAGIOPercentLibOnly
  & \easyScryptRAGNoImportCompPercent/\easyScryptRAGNoImportCompPercentLibOnly
  & \easyScryptRAGNoImportIOPercent/\easyScryptRAGNoImportIOPercentLibOnly \\

\simpleScrypt~\cite{simplescrypt}
  & \simpleScryptRAGLoad
  & \simpleScryptCompPercent/\simpleScryptCompPercentLibOnly
  & \simpleScryptIOPercent/\simpleScryptIOPercentLibOnly
  & \simpleScryptNoRAGCompPercent/\simpleScryptNoRAGCompPercentLibOnly
  & \simpleScryptNoRAGIOPercent/\simpleScryptNoRAGIOPercentLibOnly
  & \simpleScryptRAGNoImportCompPercent/\simpleScryptRAGNoImportCompPercentLibOnly
  & \simpleScryptRAGNoImportIOPercent/\simpleScryptRAGNoImportIOPercentLibOnly \\

\argonTwoHashing~\cite{argon2hashing}
  & \argonTwoHashingRAGLoad
  & \argonTwoHashingCompPercent/\argonTwoHashingCompPercentLibOnly
  & \argonTwoHashingIOPercent/\argonTwoHashingIOPercentLibOnly
  & \argonTwoHashingNoRAGCompPercent/\argonTwoHashingNoRAGCompPercentLibOnly
  & \argonTwoHashingNoRAGIOPercent/\argonTwoHashingNoRAGIOPercentLibOnly
  & \argonTwoHashingRAGNoImportCompPercent/\argonTwoHashingRAGNoImportCompPercentLibOnly
  & \argonTwoHashingRAGNoImportIOPercent/\argonTwoHashingRAGNoImportIOPercentLibOnly \\

\duration~\cite{duration}
  & \durationRAGLoad
  & \durationCompPercent/\durationCompPercentLibOnly
  & \durationIOPercent/\durationIOPercentLibOnly
  & \durationNoRAGCompPercent/\durationNoRAGCompPercentLibOnly
  & \durationNoRAGIOPercent/\durationNoRAGIOPercentLibOnly
  & \durationRAGNoImportCompPercent/\durationRAGNoImportCompPercentLibOnly
  & \durationRAGNoImportIOPercent/\durationRAGNoImportIOPercentLibOnly \\
  
\age~\cite{age}
  & \ageRAGLoad
  & \ageCompPercent/\ageCompPercentLibOnly
  & \ageIOPercent/\ageIOPercentLibOnly
  & \ageNoRAGCompPercent/\ageNoRAGCompPercentLibOnly
  & \ageNoRAGIOPercent/\ageNoRAGIOPercentLibOnly
  & \ageRAGNoImportCompPercent/\ageRAGNoImportCompPercentLibOnly
  & \ageRAGNoImportIOPercent/\ageRAGNoImportIOPercentLibOnly \\

\midrule
\textbf{Average}
  & \avgRAGLoad
  & \avgFullComp/\avgFullCompLibOnly
  & \avgFullIO/\avgFullIOLibOnly
  & \avgNoRAGComp/\avgNoRAGCompLibOnly
  & \avgNoRAGIO/\avgNoRAGIOLibOnly
  & \avgNoImportComp/\avgNoImportCompLibOnly
  & \avgNoImportIO/\avgNoImportIOLibOnly \\
  
\bottomrule
\end{tabular}
}
\label{tab:benchmarkResults}
\end{table}

In this section, we present our results for the following research questions.
\begin{itemize}
    \item \textbf{RQ1: (Overall Effectiveness)} Can our approach translate entire projects that make significant use of external libraries? Can our approach validate functional equivalence of functions using opaque, external library types as arguments?
    \item \textbf{RQ2: (Ablation)} How much do our RAG and Rust import resolution approaches contribute to translation success?
    \item \textbf{RQ3: (Baseline)} How does our approach compare against \oxidizer{}~\cite{zhang_scalable_2025}?
\end{itemize}

\subsection{Experimental Setup}
\textbf{Implementation.}
We build our approach on top of the tool \oxidizer~\cite{zhang_scalable_2025}, a state-of-the-art code translation tool. Our modifications add an additional 1580, 1245, and 2585 lines of Go, Rust, and Python to \oxidizer's implementation. Our new implementation, called \liboxidizer, is used the same way as \oxidizer: it takes as input a Go project and several hyperparameters that control retry budgets, and it outputs a Rust translation of the original project. 

\textbf{LLMs \& Hyperparameters.}
We use the same hyperparameter settings as in the original \oxidizer{}
work~\cite{zhang_scalable_2025} for all experiments.  We use Claude Sonnet
3.5~\cite{anthropic_claude35_2024} as the LLM for all experiments, which
performs similarly to other state-of-the-art LLMs in coding tasks, thus we
believe our approach would work with other LLMs as well.  To support
reproducibility, our implementation also supports replaying logged LLM
responses without querying a real LLM.

\textbf{Benchmarks.}
We evaluate our approach on a curated set of real-world \go{Go} benchmarks
drawn from GitHub. Our selection criteria are as follows:
(1)~the benchmarks expose non-trivial library APIs, resulting in a large
candidate search space for retrieval-augmented generation;
(2)~they exercise \emph{distinct dependency domains}, including cryptographic
key derivation, memory-hard password hashing, and time and duration parsing;
and (3)~they are actively maintained and used in production systems.
%
This diversity of dependency domains enables us to study how our approach
behaves across fundamentally different API and dependency structures,
ranging from security-critical, parameter-rich cryptographic libraries to
domain-specific parsing libraries with intricate control flow and type-level
constraints.



Concretely, our benchmark suite consists of the following libraries.
Table~\ref{tab:benchmarks} summarises the benchmarks, their dependency domains,
and popularity (stars and forks).

\begin{itemize}
    \item~\pbkdft~\cite{pbkdft}:
    a cryptographic key-derivation library implementing the \texttt{PBKDF2} algorithm, used to derive secure keys from passwords via salts and iteration counts. This benchmark exercises standard cryptographic APIs and parameterised security primitives.

    \item~\easyScrypt~\cite{easyscrypt} and \simpleScrypt~\cite{simplescrypt}:
    memory-hard key-derivation libraries providing \texttt{scrypt}-based password hashing. 
    These benchmarks stress translation of security-critical code with non-trivial parameter validation, memory-cost configuration, and performance-sensitive APIs.

    \item~\argonTwoHashing~\cite{argon2hashing}:
    a password-hashing library based on the \texttt{Argon2} algorithm representing modern, memory-hard password hashing with structured parameter types and strong API-level constraints.

    \item~\duration~\cite{duration}:
    a time-processing library for parsing and formatting human-readable duration strings exercising domain-specific parsing logic and interactions between numeric, temporal, and string-processing APIs.

    \item~\age~\cite{age}:
     a simple, modern, and secure file encryption tool featuring small explicit keys and post-quantum support. We extracted functions and types with complex library dependencies from this library, while dropping fragments with features unsupported by \oxidizer{} (e.g., involving file I/O).
\end{itemize}

\subsection{Results}

We run \liboxidizer{} on each of our benchmarks, and report results in Table~\ref{tab:benchmarkResults} under the column \textbf{Full}. We report two key metrics: compilation rate under the column \textbf{\%~Comp.} and I/O equivalence rate under the column \textbf{\%~Equiv.}. Both of these metrics are computed as the percentage of functions that are compilable or equivalent. We report these metrics for all functions (\textbf{full}), and for the subset of functions that use library dependencies (\textbf{dep}).

\paragraph{RQ1 (Overall Effectiveness).}
The results in Table~\ref{tab:benchmarkResults} show that \liboxidizer{} indeed can effectively translate projects with external libraries, and can validate functional equivalence of functions that use external library types as arguments. \liboxidizer{} successfully translates and validates 100\% of the functions in all of our benchmarks.

\paragraph{RQ2 (Ablation).}
Next, we assess the usefulness of both the knowledge base and the import resolution in our RAG approach. We run \liboxidizer{} in two modes: one with both the knowledge base look-ups and import resolution disabled, and one with just knowledge base look-ups enabled. The results are shown in columns \textbf{No RAG}, and \textbf{No Imports}, respectively, in Table~\ref{tab:benchmarkResults}. The results show that, without both  components of our RAG approach, the LLM fails to produce compiling (and equivalent code) on functions with library dependencies, and the overall compilation and equivalence rates drop significantly. These findings demonstrate that import resolution is a necessary component in our pipeline, and that API mapping without call-enabling imports does not suffice in practice.

\paragraph{RQ3 (Baseline).} 
\oxidizer~hard-codes target crates for translation; however, this provides limited help with complex libraries---performing identically to the \textbf{No RAG} ablation and achieving \(0\%\) compilation and equivalence success on functions with external \api{}s.

%% file: sections/related_work.tex
\section{Related Work}
\textbf{Code Translation Validation.}
Code translation validation was originally studied in the context of compilers~\cite{necula_translation_2000,pnueli_translation_1998}, and only recently has received significant attention in our domain~\cite{szafraniec_code_2023,roziere_leveraging_2022,roziere_unsupervised_2020,pan_lost_2024}, namely source-to-source translation. Early works proposed rule-based translation systems~\cite{noauthor_sharpen_nodate,noauthor_c_nodate,malabarba_mohca-java_1999,buddrus_cappuccinoc_1998}, many focused on C to Rust~\cite{immunant_c2rust_nodate,zhang_ownership_2023,sharp_corrode_nodate,ling_rust_2022,hong_dont_2024,fromherz_compiling_2024,emre_translating_2021}, which are theoretically correct by construction, but generally produce unidiomatic translations or target specialized domains~\cite{zhang_heterogen_2022,zhan_verified_2024,wang_user-customizable_2023,qiu_tenspiler_2024,kamil_verified_2016,ahmad_automatically_2019}. More recent work uses LLMs to propose a candidate translation, which is then combined with an equivalence validation system. Validation us done at one of two-levels: system-level~\cite{li_adversarial_2025} or function-level~\cite{zhang_scalable_2025,yang_vert_2025,shetty_syzygy_2024,zhou_llm-driven_2025,ibrahimzada_alphatrans_2025,eniser_automatically_2024,bhatia_verified_2024,abid_gluetest_2024,ke_advancing_2025,bai2025rustassure}. System level validation is done by executing a textual interface, like a command-line interface. Works on system-level validation assume textual interfaces for the source and translation are identical, and thus do not address the challenges we solve, moreover they cannot validate library translations that do not provide textual interfaces. Works on function-level generally assume function input types are structurally similar. SACTOR~\cite{zhou_llm-driven_2025} is the only work that does not make this assumption, however their approach does not handle opaque types, and depends on hand-written rules specific to C and Rust.
\\
\textbf{Unvalidated Code Translation.}
Many works in code translation focus only on generating the best candidate translation, but do not propose automated validation systems~\cite{xu_optimizing_2025,wang_program_2025,shiraishi_smartc2rust_2026,nitin_c2saferrust_2025,hong_type-migrating_2025,cai_rustmap_2025,zhang_multilingual_2023,pan_lost_2024,dehghan2025translating,wang2025evoc2rust,wang_repotransbench_2024}. These works instead use hand-written tests to evaluate the correctness of the translation. They are orthogonal to our work.
\textbf{LLM-Based Code Generation with External Libraries}
A recent study shows that LLM performance in code translation degrades when the source code makes heavy use of external libraries~\cite{gu_effectiveness_2025}. A parallel work to ours proposes a library mapping approach similar to ours, which uses documentation to find equivalent external libraries in the target language~\cite{ke_advancing_2025}, though they do not address the challenges specific to Rust, like we do. LLM-based code generation with external libraries has been studied outside the context of code translation as well~\cite{du_evaluating_2024}, however they are not attempting to build a library mapping.
\\
\textbf{Cross-Language Differential Testing and Verification.}
Much work has been done on differential testing and verification in domains besides code translation~\cite{petsios_nezha_2017,palikareva_shadow_2016,noller_hydiff_2020,nilizadeh_diffuzz_2019,garzella_leveraging_2020}, however they assume the interfaces of the two entities under test are identical, so they do not address one of our core challenges.

%% file: sections/conclusions.tex
\section{Conclusions}
\label{sec:conclusions}

We proposed a library-aware, RAG-guided framework for Go to Rust
translation.  It maps external Go APIs to Rust using a curated knowledge
base of public, crate-root-reachable items, then validates results in two
steps: cross-language round trips followed by differential I/O.  On projects
with external dependencies, this improved both compilation and equivalence
success rate (approximate 2x on average) by enabling translation for
functions with external dependencies.

%% file: sections/api-index-construction.tex
\section{An $\itemsmap$ Construction Algorithm for \rustlang}
\label{app:itemsmap-rust}

We first formalize the notions of valid import path and re-exports.
\paragraph{\textbf{Paths.}}
A path  $\elem{P}$ is a fully-qualified name used to reference an \api{} item within a
crate. Intuitively, a path begins at the crate root module and traverses zero
or more nested submodules before terminating at an item  $\elem{i}$.
However, not every syntactically well-formed path is a \emph{valid import path}.
In Rust, each module has an associated visibility, and items contained within a
private module are not accessible to external clients. Consequently, a path is
invalid whenever it passes through a module that is not visible from the crate
root under Rust's privacy rules.

\begin{definition}[Valid import path]
Let \(m_0\) be the root module of a crate. A~syntactic path
\[
p = m_0 :: m_1 :: \cdots :: m_k :: i
\]
is a \textbf{valid import path} if every intermediate module on the path is
publicly visible, i.e.
\[
\mathsf{Valid}(p) {:=} \forall j \in \{1,\ldots,k\}.\ \mathsf{Public}(m_j).
\]
\end{definition}

\paragraph{\textbf{Re-exports.}}
Rust additionally allows developers to expose items through alternative paths
using re-exports. A re-export statement \rustcode{pub use} \(p\) introduces an
alias for the item referenced by \(p\), making that item accessible from the
re-exporting module under a new path. Re-exports are widely used to define clean
public APIs that hide internal module hierarchies, and to surface selected items
from transitive dependencies as part of a crate's public interface.

\begin{definition}[Re-export resolution]
Let \(m\) be a module and let \(p\) be a path. If \(p\) resolves to an item \(i\),
written \(p \Downarrow i\), then a re-export declaration
\(\rustcode{pub use}\ p\) in \(m\) induces an alias path \(m::i'\) such that
\[
(p \Downarrow i) \;\Rightarrow\; (m::i' \leadsto i).
\]
\end{definition}

\paragraph{Constructing $\itemsmap$.}
For example, in~\autoref{fig:ed255192curve25519} the \api{} \rustcode{SigningKey}
is defined in the internal module \rustcode{ed25519\_dalek::signing}. Since
\rustcode{signing} is private, the defining path is not a valid import path for
crate users. The crate therefore re-exports \rustcode{SigningKey} at the root,
making \rustcode{ed25519\_dalek::SigningKey} a valid import path. Re-exports are
also commonly used to propagate APIs from dependencies. For instance,
\(\rustcode{Sha512::new}\) is available through the trait \rustcode{sha2::Digest},
which is defined in the transitive dependency \rustcrate{digest}.

To construct the $\itemsmap$ for a crate, we must also analyse its transitive dependencies, as \rustcode{pub use} re-exports can expose dependency-defined \api{}s as part of the crate's public interface. As shown in~\autoref{alg:main-entry}, we compute $\itemsmap$s for dependencies in post-order, ensuring they are available before processing the current crate. These precomputed maps are passed into the extraction of the target crate's $\itemsmap$, which is returned as the final result.

The procedure for computing a per-crate $\itemsmap$ is detailed in~\autoref{alg:extract-main}. The main entry point, \textsc{ComputeItemsMap}, first invokes \textsc{traverseModule} to construct an initial $\itemsmap$ that may contain placeholder import paths for method \api{}s. It then performs a revision pass to resolve these placeholders into concrete import paths, yielding a fully resolved $\itemsmap$.

\paragraph{Sub-routine for traversing a crate}
The recursive subroutine \textsc{traverseModule} enumerates all publicly
reachable modules and items within a crate. It maintains the current module path
prefix \(\Pi\), representing the sequence of enclosing modules from the crate
root to the current module, and uses this prefix to construct import paths for
each discovered \api{}.

For an external re-export of the form
\(\rustcode{pub use}\ n::\ldots::i\) where \(n \neq m\) and \(m\) denotes the
current crate, the root \(n\) refers to a crate in the transitive dependency
set. In this case, we retrieve the \api{} entry for \(i\) from the precomputed
\textit{CratesMap} and record it under the re-exported path \(\Pi'::i\),
reflecting its publicly exposed location in the current crate.

For an internal re-export of the form \(\rustcode{pub use}\ m::\ldots::i\), where the referenced path originates from
the current crate root, we treat \(i\) as locally defined within the current
module. This allows internal re-exports to be processed uniformly with
module-local public items during traversal.


Finally, for each public item in the current module---together with the items
introduced via internal re-exports---we either extract its associated
documentation (for types and traits) or, if the item is itself a submodule,
recursively invoke \textsc{traverseModule} to continue the traversal.

\paragraph{Sub-routine for extracting \api{}} We extract traits, types, and methods as outlined in~\autoref{alg:extract}. For method \api{}s, although the input path $\Pi$ identifies the type with which these methods are associated, it cannot be used directly to compute their correct import paths. This limitation arises from Rust’s trait system: the use of a trait method is only permitted when the defining trait $\mathbb{T}$ is in scope, which requires importing $\mathbb{T}$ itself. Consequently, the appropriate import path for a method \api{} can only be determined after all traits within the library have been identified and their import paths established. To accommodate this dependency, we initially record method \api{}s with a placeholder import path corresponding to the defining trait $\mathbb{T}$, to be resolved in a later stage.

\paragraph{Final stage of revision} After constructing an initial $\itemsmap$ that may contain placeholder import paths for method \api{}s, the main procedure \textsc{ComputeItemsMap} performs a final revision pass. During this pass, each placeholder trait reference is replaced with the concrete import path computed for the corresponding trait, yielding a fully resolved $\itemsmap$ for the target library.

\begin{algorithm}
\caption{Main API Extraction Procedure}\label{alg:main-entry}
\begin{algorithmic}
\footnotesize
\Require $\elem{M}$ The root module
\Ensure $\itemsmap$ Documentation and available paths for public APIs
\Function{ExtractCrate}{$\elem{M}$}
\State $\textit{CratesMap} \gets \emptyset$
\For{$\elem{N} \in \Call{transitiveDependencies}{\elem{M}}$ in post order}
    \State $\textit{CratesMap}\left[\elem{N}\right] \gets$ \Call{ComputeItemsMap}{$\emptyset$, $\elem{N}$, $\textit{CratesMap}$}
\EndFor
\State \Return $\textit{CratesMap}\left[\elem{M}\right]$

\EndFunction
\end{algorithmic}
\end{algorithm}

\begin{algorithm}
\caption{Computing Items Map for a Module}\label{alg:extract-main}
\begin{algorithmic}[1]
\Require \\
$\elem{M}$ Input module\\
$\textit{CratesMap}$ Mapping third party crates to their items map
\Ensure $\itemsmap$ The set of public APIs

\Function{ComputeItemsMap}{$\elem{M}$, $\textit{CratesMap}$}
    \State $\itemsmap \gets $ \Call{traverseModule}{$\emptyset$, $\elem{M}$, $\textit{CratesMap}$}
    \For{$f \in \textit{functions}\left(\itemsmap\right)$}
        \State $\text{Doc}, \mathbb{T} \gets \itemsmap[f]$
        \If{$\mathbb{T} \in \textit{traits}\left(\itemsmap\right)$}
            \State $\_, p \gets \itemsmap[\mathbb{T}]$
            \State $\itemsmap[f] \gets \text{Doc}, p$
        \Else
            \State \Call{Delete}{$\itemsmap$, $f$}
        \EndIf
    \EndFor
    \State \Return $\itemsmap$
\EndFunction

\Function{traverseModule}{$\Pi$, $\elem{M}$, $\textit{CratesMap}$}

\State $\Pi^\prime \gets \Pi::\elem{M}$
\State $\itemsmap \gets \emptyset$

\For{$\rustcode{pub use}\ n::...::i \in \Call{ExternalRe-exports}{\elem{M}}$}
    \State $\text{Doc}, \_ \gets \textit{CratesMap}[\elem{N}][\elem{I}]$
    \State $\itemsmap[\elem{i}] \gets \text{Doc}, \Pi^\prime::\elem{I}$
\EndFor

\State $\textit{InteralReexported} \gets \emptyset$
\For{$\rustcode{pub use}\ m::...::i \in \Call{InternalRe-exports}{\elem{M}}$}
    \State \Call{Append}{$\textit{InteralReexported}$, $\elem{I}$}
\EndFor

\For{$\elem{I} \in \Call{PublicItems}{\elem{M}} \cup \textit{InteralReexported}$}
    \If{$\elem{I}$ is Module}
    \State $\itemsmap \gets \itemsmap\ \cup$ \Call{traverseModule}{$\Pi^\prime$, $\elem{I}$, $\textit{CratesMap}$}
\ElsIf{$\elem{I}$ is Trait}
    \State $\itemsmap \gets \itemsmap\ \cup$ \Call{extractTrait}{$\Pi^\prime$, $\elem{I}$}
\ElsIf{$\elem{I}$ is Type}
    \State $\itemsmap \gets \itemsmap\ \cup$ \Call{extractType}{$\Pi^\prime$, $\elem{I}$}
\EndIf
\EndFor

\State \Return $\itemsmap$
\EndFunction

\end{algorithmic}
\end{algorithm}

\begin{algorithm}
\caption{Extracting Functions, Types and Traits}\label{alg:extract}
\begin{algorithmic}
\Function{extractType}{$\Pi$, $\elem{T}$}
\State $\itemsmap \gets \left\{\elem{T} \mapsto \left(\text{Doc}(\elem{T}), \Pi::\elem{T}\right)\right\}$

\For{$\left(\rustcode{impl }\mathbb{T}\rustcode{ for } \elem{T} \left\{\overline{f}\right\}\right) \in \textit{ImplBlocks}(\elem{T})$}
\For{$f \in \overline{f}$}
    \State $\itemsmap[f] \gets \text{Doc}\left(f\right), \mathbb{T}$
\EndFor
\EndFor
\EndFunction

\Function{extractTrait}{$\Pi$, $\mathbb{T}$}
\State \Return $\left\{\mathbb{T} \mapsto \left(\text{Doc}(\mathbb{T}), \Pi::\mathbb{T}\right)\right\}$
\EndFunction
\end{algorithmic}
\end{algorithm}

\paragraph{\textbf{Soundness of Path Generation.}}
Our $\itemsmap$ construction enumerates candidate import paths by traversing the
public module tree and collecting all paths to reachable items, additionally
including paths introduced via \rustcode{pub use} re-exports.

\begin{lemma}[Validity of generated import paths]
Every path \(p\) produced by our $\itemsmap$ construction satisfies
\(\mathsf{Valid}(p)\).
\end{lemma}

\begin{proof}[sketch]
A path is produced in one of two ways.
(1) \emph{Direct paths:} These are obtained by traversing the module hierarchy
starting from the crate root and following only publicly visible submodules.
Thus, every prefix module on such a path is public by construction, so
\(\mathsf{Valid}(p)\) holds.
(2) \emph{Re-exported paths:} A re-export path is introduced only by a
\(\rustcode{pub use}\) declaration occurring in some publicly reachable module.
Since the re-export is itself public, the resulting alias path is externally
visible from the root, hence \(\mathsf{Valid}(p)\) holds.
\end{proof}

%% file: sections/appendix.tex
\section{Translation Validation Algorithm}

\autoref{alg:establish-roundtripping} summarises the translation procedure.

\begin{algorithm}
\caption{Establish cross-language interoperability}\label{alg:establish-roundtripping}
\begin{algorithmic}
\Require $T_{\go\golang}$, $T_{\rust\rustlang}$, and \golang values $\overline{v}$

\Ensure $T_{\mathrm{Proto}}$, $\mathcal{A}_{T_{\mathrm{Lang}}}$ and $\mathcal{A}_{T_{\mathrm{Lang}}}^{-}$ where $\mathrm{Lang} \in \left\{\go\golang, \rust\rustlang\right\}$
\Function{GenInteroperators}{$T_{\go\golang}$, $T_{\rust\rustlang}$}

\State \textbf{/*Generate a Proto schema, Go glue code and validate Go $\leftrightarrow$ Proto*/}
\While{$\mathit{true}$}
    \State $T_{\mathrm{Proto}} \gets \Call{GenProtoSchema}{T_{\go\golang}}$
    
    \State $\mathcal{A}_{T_{\go\golang}}, \mathcal{A}_{T_{\go\golang}}^{-} \gets \Call{GenHighLevelGlueCode}{T_{\go\golang},  T_{\mathrm{Proto}}}$
    
    \If{$\Call{CheckRoundtrip}{T_{\go\golang},  T_{\mathrm{Proto}}, \mathcal{A}_{T_{\go\golang}}, \mathcal{A}_{T_{\go\golang}}^{-}, \overline{v}}$}
        \State {\bf break}
    
    \EndIf
\EndWhile

\State \textbf{/*Generate Rust glue code and validate Go $\leftrightarrow$ Rust*/}
\While{$\mathit{true}$}
\State $\mathcal{A}_{T_{\rust\rustlang}}, \mathcal{A}_{T_{\rust\rustlang}}^{-} \gets \Call{GenHighLevelGlueCode}{T_{\rust\rustlang},  T_{\mathrm{Proto}}}$

\State $\mathcal{A}^{-}_{T}, \mathcal{A}_{T} \gets \mathcal{A}_{T_{\go\golang}}^{-} \circ \mathcal{A}_{T_{\rust\rustlang}}, \mathcal{A}_{T_{\rust\rustlang}}^{-} \circ \mathcal{A}_{T_{\go\golang}}$

\If{$\Call{CheckRoundtrip}{T_{\go\golang},  T_{\rust\rustlang}, \mathcal{A}_{T}, \mathcal{A}^{-}_{T}, \overline{v}}$}
    \State {\bf break}

\EndIf

\EndWhile

\EndFunction
\end{algorithmic}
\end{algorithm}